\documentclass[conference,10pt,a4paper]{IEEEtran}
\usepackage{graphicx}
\usepackage{amsthm}
\usepackage{empheq}
\usepackage{amsfonts}
\usepackage{amsmath,amsthm,amssymb}
\usepackage{graphicx,subfigure}
\usepackage{color}
\usepackage{txfonts}
\usepackage{booktabs,array}
\usepackage{multirow}
\usepackage{url}
\usepackage{tikz}
\usepackage{booktabs}
\usepackage[utf8]{inputenc}
\DeclareMathOperator{\E}{\mathbb{E}}

\renewcommand{\eqref}[1]{(\ref{#1})}
\graphicspath{{figures/}}
\newtheorem{Lemma}{Lemma}
\usepackage{epstopdf}
\usepackage{steinmetz}

\usepackage[caption=false,font=footnotesize,labelfont=sf,textfont=sf]{subfig}
\usepackage{flexisym}
\usepackage[T1]{fontenc}

\begin{document}
	\title{Statistical QoS Analysis of Full Duplex and Half Duplex 
	Heterogeneous Cellular Networks}
\author{\IEEEauthorblockN{Alireza Sadeghi$^{*}$, Michele Luvisotto$^\#$, 
Farshad Lahouti$^{+}$, Stefano Vitturi$^\dag$, Michele Zorzi$^\#$}

	\IEEEauthorblockA{$^*$Electrical and Computer Engineering Department, University of Minnesota, USA\\ }
	\IEEEauthorblockA{$^\#$Department of Information Engineering, University of Padova, Italy}
	\IEEEauthorblockA{$^+$Electrical Engineering Department, California 
	Institute of Technology, USA\\ }
	\IEEEauthorblockA{$^\dag$CNR-IEIIT,	National Research Council of Italy, 
	Padova, Italy\\ }
	{sadeg012@umn.edu, \{luvisott, zorzi, vitturi\}@dei.unipd.it, lahouti@caltech.edu}}	

\maketitle
\begin{abstract}
In this paper, statistical Quality of Service provisioning in next generation 
heterogeneous mobile cellular networks is investigated. To this aim, any active 
entity of the cellular network is regarded as a queuing system, whose 
statistical QoS requirements depend on the specific application. In this 
context, by quantifying the performance in terms of effective capacity, we 
introduce a lower bound for the system performance that facilitates an 
efficient analysis. We exploit this analytical framework to give insights about 
the possible improvement of the statistical QoS experienced by the users if the 
current heterogeneous cellular network architecture migrates from a Half 
Duplex to a Full Duplex mode of operation. Numerical results and analysis are 
provided, where the network is modeled as a Matérn point processes with a hard 
core distance. The results demonstrate the accuracy and computational 
efficiency of the proposed scheme, especially in large scale wireless systems.
\end{abstract}

\section{Introduction}
\label{sec:introduction}
The ever increasing demand for mobile data traffic continues with the advent of smart phones, tablets, mobile 
routers, and cellular M2M devices. This is accompanied by user
behavioral changes from web browsing towards video streaming, social 
networking, and online gaming with distinct QoS requirements 
\cite{ericsson2014}. To handle this challenging scenario, researchers are 
examining different enabling technologies for 5G, including mmWave 
communications for wider bandwidth, extreme densification of the network via 
low 
power base stations (known as heterogenous networks), the use of large--scale 
antenna systems (known as Massive MIMO), and wireless Full Duplex entities 
\cite{andrews20145G}. 

The new cellular architecture known as Heterogeneous Cellular Networks (HCNs) 
refers to a scenario in which the macro cellular network is overlaid by 
heterogeneous low--power base stations (BSs). 
Such low power BSs have small coverage areas and are characterized by 
their own transmit power and named accordingly as micro, pico and femto cells. 
They are used to increase the capacity of the network while 
eliminating coverage holes \cite{HetNetfromTheorytoPractic}.

The Full Duplex (FD) radio technology enhances spectrum 
efficiency by enabling a node to transmit and receive in the same 
frequency band at the same time. This new emerging technology has the potential 
to double the physical layer capacity and enhance the performance even more, 
when higher layer protocols are redesigned accordingly \cite{InbandSabharwal}.

Due to the hurdles of canceling self--interference (SI) in FD devices via 
active and passive suppression mechanisms, FD operations are more reliable in 
low power wireless nodes. For instance in \cite{bharadia2013full} the authors 
have implemented an FD WiFi radio operating 
in an unlicensed frequency band with 20~dBm transmit power while the same trend 
is followed in other works like \cite{DuarteExperimentalCharach} where the 
maximum transmit power is 15~dBm. All these implementations suggest FD 
technology as a very good candidate to be used in the low power BSs deployed within 
HCNs. 
Moreover, the increased spectral efficiency of the FD systems, combined with 
that of HCNs, provides another strong motivation in attempting to analyze an FD 
HCN. 

From another perspective, next generation mobile networks (5G) will aim not 
only to increase the network capacity but also to enhance several other 
performance metrics, including lower latency, seamless connectivity, and increased mobility \cite{andrews20145G}. These enhancements can be generally 
referred to as an improvement in the Quality of Service (QoS) experienced by 
the network entities. According to a forecast by Ericsson mobility report, 
approximately 55 percent of all the mobile data traffic in 2020 will account 
for mobile video traffic while another 15 percent will account for social 
networking \cite{ericsson2014}. These multimedia services require a
bounded delay. Generally the delay requirements of time sensitive services in 5G will vary extremely, from milliseconds to a few seconds \cite{ZhangHetQoS}. 
Consequently, the analysis of statistical QoS in HCNs will become extremely important in 
the near future. 

The objective of this paper is to analyze and compare FD and HD HCNs 
in provisioning statistical QoS guarantees to the users in the network. The QoS is assessed statistically in terms of Effective Capacity (EC) as the maximum throughput under a delay constraint \cite{DapengQoS}.
 
Our goal is to provide insights on possible improvements in the QoS experience 
of end users if the current architecture migrates from conventional HD to FD. 
To this end, we propose a lower bound for the EC which greatly 
reduces the complexity of the analysis while tightly approximating the system performance, 
especially in large scale systems. Our results will be validated through numerical 
simulations.

The rest of the paper is organized as follows. Some basic explanations on FD, 
statistical QoS provisioning, and stochastic geometry are provided in 
Section~\ref{Preliminaries}. The system model is 
described in Section~\ref{System Model}. 
Section~\ref{sec:theoretical} presents the proposed lower bound for the system 
performance and the corresponding theoretical analysis, whose results are validated 
through simulations in Section~\ref{sec:simulations}. Finally, 
Section~\ref{sec:conclusions} concludes this paper.
%\vspace{-.5cm}
\section{Preliminaries}
\label{Preliminaries}
\subsection{Full Duplex}
In--band Full Duplex (IBFD) devices are capable of transmitting and receiving 
data in the same frequency band at the same time. In traditional wireless 
terminals, the ratio of SI power with respect to that of the received intended 
signal is very high, making any reception infeasible while a transmission is 
ongoing. To overcome this issue, FD terminals are equipped with active and 
passive cancellation mechanisms to suppress their own SI in the received signal \cite{InbandSabharwal}. 
However, in practice, because of the many imperfections in transceiver operations, 
full cancellation of the SI signal is not possible. Therefore, some residual 
self--interference (RSI) always remains after all cancellation steps and
results in a degraded system performance. 

The RSI signal represents the main obstacle for a perfect FD communication and, 
similar to noise, is essentially uncorrelated with the original transmitted 
signal. We model the RSI signal at the FD transceiver as a complex Gaussian 
random variable \cite{RamirezOptimal}
\begin{equation}
\mathrm{RSI} \sim \mathbb{CN} \left( {0,\eta {P^\kappa }} \right),
\label{eq}
\end{equation}
where $P$ is the transmit power, while $\eta$ and $\kappa$ are parameters to 
model the SI cancellation performance.
Specifically, $\eta$ is the linear SI cancellation parameter, while $\kappa$ 
models non--linear 
SI cancellation, $0 \le \eta,\kappa  \le 1$. 
When no SI cancellation is performed $\eta , \kappa=1$, while $\eta=0$ 
represents the ideal case of perfect SI cancellation.

\subsection{Statistical QoS guarantees}
\label{subsec:qos_guarantees}

Real--time multimedia services like video streaming require bounded delays. In 
this context, a received packet that violates its delay bound requirement is 
considered useless and discarded. 
Due to the wireless nature of the access links in a mobile cellular network, 
providing deterministic delay bound guarantees is not possible. 
Thus, the concept of EC, defined as the maximum throughput under a given delay 
constraint, has been used to analyze multimedia wireless systems 
\cite{zhang2006cross}. 
Any active entity in a cellular network can be regarded as a queueing system: 
it 
generates packets according to an arrival process, stores them in a queue and 
transmits them according to a service process.
For stationary and ergodic arrival and service processes, the probability that 
the queue size, $Q$, exceeds a certain threshold, $B$, decays exponentially 
fast as the threshold increases \cite{zhang2006cross}, i.e.,

\begin{equation}
\Pr \left\{ {Q > B} \right\} \sim {e^{ - \theta B}}  \qquad  \textrm{as} \, B 
\rightarrow \infty ,
\label{eq1}
\end{equation}
where $\theta$ denotes the decaying rate of the QoS violation probability. The 
smaller $\theta$, the looser the QoS requirement.

Define the service provided by the channel until time 
slot $t$ as 
\begin{equation}
C\left( {0,t} \right) = \sum\limits_{k = 1}^t {R\left[ k \right]},
\label{eq3}
\end{equation}
where $R[k]$ denotes the number of bits served in time slot $k$. 
The effective capacity of the channel is defined as \cite{DapengQoS}
\begin{equation}
{\rm{EC}}\left( \theta  \right) =  - \frac{{{\Lambda _C}\left( { - \theta } \right)}}{\theta }
\label{eq4}
\end{equation}
where 
$\Lambda_C\left(-\theta\right)=\lim\limits_{t\to\infty}\frac{1}{t}\log\E\left\{
e^{-\theta C(0,t)}\right\}$
is the G\"artner - Ellis limit of the service process $C(0,t)$. 

If the instantaneous service process, $R[k]$, is independent in time, 
EC can be simplified to 

\begin{equation}
{\rm{EC}}\left( \theta  \right) =  - \frac{1}{\theta }\log {\E\left \{ {{e^{ - 
\theta R[k]}}} \right \}} 
\label{eq5}
\end{equation}

It is worth mentioning that, for $\theta \rightarrow 0$, the EC tends to the 
average service rate \cite{zhang2006cross}.

\subsection{Stochastic Geometry}
\label{Stochastic Geometry}

Stochastic geometry is a powerful mathematical tool that has recently been 
proposed to model and analyze the performance of wireless networks 
\cite{StochGeoEk}. 
In particular, Poisson Point Processes (PPPs) have been vastly used to model 
the positions of the network entities in HCNs. This approach has enabled the 
study of realistic scenarios where the BSs are not placed on a hexagonal grid 
but are instead spread randomly in the network. The use of a PPP to model the 
system has become popular because of its tractability and its ability to give 
simple expressions for some network performance metrics like coverage 
probabilities and mean transmission rates \cite{StochGeoEk}, 
\cite{DhillonkHetNet}, and \cite{AndrewsPPP}. 
However, in a real cellular network, the adoption of a simple PPP to model the 
locations of the BSs does not capture an important characteristic of the 
network, namely the constraints on the minimum distance between any two 
BSs or UEs imposed by the MAC layer, network planning or spectrum access policies 
\cite{StochGeoEk}. 
Consequently, spatial correlation among different network entities should be 
taken into account. According to these considerations, a repulsive point 
process with a hard core distance such as the Mat\'ern hard core point process 
(HCPP), despite its higher complexity, represents a better candidate to model a HCN 
compared to a simple PPP \cite{StochGeoEk}.

\section{System Model}
\label{System Model}

We refer to Fig.~\ref{fig1} as our system model in both HD and FD scenarios. 
When the system is HD a conventional HCN is assumed, while in the FD case we 
consider a completely FD HCN where all the network entities are assumed to be 
(imperfect) FD devices. 
In our model we consider a circular macro cell, overlaid by different tiers of small cells, each with its own 
characteristics including transmit power, path loss exponent, and coverage 
range. 
Each tier is assumed to have a circular coverage area provided by an 
omnidirectional antenna to serve any user 
within its coverage range. 

In addition, we assume a Mat\'ern Point Process with hard core distance to 
model the location of non--overlapping small 
cells and the distribution of the user equipments (UEs) within each small cell. In our FD HCN system model, the nodes communicate in bidirectional FD mode, as depicted in Fig. \ref{fig1}. The small cells are assumed to use out of band resources 
like fiber optics, wire, or microwave links for backhauling.  

In our analysis, we assume that the positions of the BSs are known by the 
network operator. This assumption is not far from reality since, when the 
network operator wants to give service to a BS, the location of the 
BS must be communicated to the operator. The locations of the users of 
each BS are assumed to be uniformly distributed in the coverage area of that 
BS, resulting in a Mat\'ern PP but with known locations for the cluster 
heads.

In the FD scenario, a UE in the network experiences three different types of 
interference: (1) RSI, due to concurrent transmission and reception in the same 
frequency band at the same time; (2) interference from BSs that are 
transmitting in the same resource blocks (RBs) in which the UE is served; and 
(3) interference from other UEs in the network that are transmitting in the 
same RBs in which the UE is served. In an HD scenario, instead, the UE will not 
face RSI and interference from other UEs. However, the interference from other 
BSs will still be present and can be even greater than in the FD case, 
according to the adopted scheduling policy. To consider the worst case, in both 
the FD and HD scenarios, we have assumed that all the BSs in the network are 
transmitting in the same RB where the UE is served.

We analyze the system on a resource block basis. In fact the scheduling 
decisions in an LTE--Advanced cellular network are made on 
a 1~ms basis and each time the scheduler in the BS grants an arbitrary 
combination of 180 kHz $\times$ 0.5 ms wide RBs to a UE based on the Buffer 
Status Report (BSR) and Channel State Information (CSI) obtained by measuring 
the reference signals in both time and frequency \cite{dahlman20134g}.
This allows the scheduler to track the variations of the channel in time and 
frequency in order to schedule resources efficiently. 
%Most importantly, this RB based analysis helps us in finding the lower bound 
%for the system performance.

\begin{figure}
	\vspace{-.5cm}
	\centering
	\scalebox{0.9}{\begin{tikzpicture}
	
	\node at (0,1.15) {\includegraphics[scale=0.3]{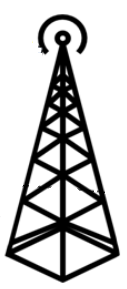}};
	\node at (1.35,-0.2) {\includegraphics[scale=0.06]{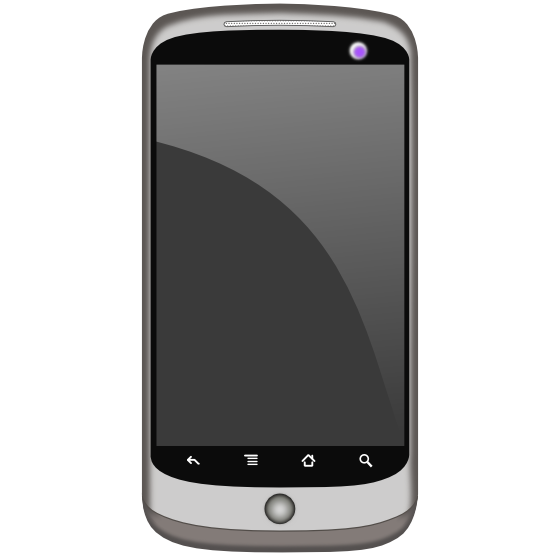}};
	\node at (4.65,-.95) {\includegraphics[scale=0.06]{ue.png}};
	\draw[draw=black,dotted] (0,0) ellipse (4.3cm and 1.95cm);
	
	\draw[draw=black,dashed] (-2.5,-0.75) ellipse (1.25cm and .45cm);
	\draw[draw=black,dashed] (-1.3,+1.25) ellipse (.85cm and .2cm);
	\draw[draw=black,dashed] (+2.25,+0.75) ellipse (1cm and .3cm);
	
	\node at (-2.5,-0.2) {\includegraphics[scale=0.15]{macro_cell.png}};
	\node at (-1.3,+1.7)  {\includegraphics[scale=0.11]{macro_cell.png}};
	\node at (+2.25,1.25) {\includegraphics[scale=0.13]{macro_cell.png}};

	\draw[draw=black,dashed] (-2.65,+.7) ellipse (0.48cm and .15cm);
	\draw[draw=black,dashed] (+1.25,+1.35) ellipse (0.45cm and .095cm);
	\draw[draw=black,dashed] (0.25,-1.1) ellipse (0.44cm and .11cm);
	\draw[draw=black,dashed] (1,-1.5) ellipse (0.47cm and .14cm);

	\node at (-2.64,1) {\includegraphics[scale=0.075]{macro_cell.png}};
	\node at (+1.2,+1.65) {\includegraphics[scale=0.07]{macro_cell.png}};
	\node at (0.23,-.78) {\includegraphics[scale=0.08]{macro_cell.png}};
	\node at (1,-1.15) {\includegraphics[scale=0.08]{macro_cell.png}};
	
	\draw[draw=black,dashed] (3.75,-1) ellipse (1.5cm and 0.4cm);
	
	\node at (3.25,-.75) {\includegraphics[scale=0.08]{macro_cell.png}};
	
	\draw[draw=black,dashed] (-1.3,+1.25) ellipse (.85cm and .2cm);
	\draw[draw=black,dashed] (+2.25,+0.75) ellipse (1cm and .3cm);	
	
	\node at (-0.2,-.45) {\scriptsize Macro};
	\node at (-3,-1) {\scriptsize Pico};
	\node at (-1.8,+1.25) {\scriptsize Pico};
	\node at (-1.8,+1.25) {\scriptsize Pico};
	\node at (+1.7,.7) {\scriptsize Pico};
	\node at (3.85,-1.55) {\scriptsize Relay};
	\node at (-3.1,0.45) {\scriptsize Femto};
	\node at (.25,-1.35) {\scriptsize Femto};
	\node at (1.05,-1.76) {\scriptsize Femto};
	\node at (1.75,1.55) {\scriptsize Femto};
	
	\draw [-latex,draw=black,dashed] (.2,2.1)--(3.15,-.42);
	\draw [-latex,draw=black] (3.35,-.42)--(4.5,-.75);
	\node at (-3.35,-.65) {\includegraphics[scale=0.06]{ue.png}};
	\node at (-1.65,-.75) {\includegraphics[scale=0.06]{ue.png}};
	
	\draw [latex-latex,draw=black] (-3.25,-0.35)--(-2.65,+.35);
	\draw [latex-latex,draw=black,dashed] (-2.35,+.35)--(-1.65,-.5);
	\draw [latex-latex,draw=black,dotted] (.2,2.1)--(1.25,0.15);
	
	\end{tikzpicture}}
\vspace{-.25cm}
	\caption{System model: an HCN with one macrocell and several LPNs.}
	\label{fig1}
\end{figure}
\vspace{-.25cm}
\subsection{Interference from UEs and BSs}
\label{subsec:interference_modelling}

We consider a Rayleigh fading, path loss dominated, AWGN channel model.
With this model, the interference at the desired UE from another entity of the 
network located at distance $x$ is given by $Ph{\left\| x 
\right\|^{- \alpha }}$ where $P$ is transmit power of the interferer, $h$ is an 
exponential random variable modeling Rayleigh fading, $h \sim \exp \left( 1 
\right)$, and $\alpha$ represents the path loss exponent. 

Considering all the above mentioned terms, the Signal--to--Interference plus 
Noise Ratio (SINR) at the desired FD UE is expressed as \cite{DhillonkHetNet}
\begin{align} 
\nonumber
\textrm{SINR} = \frac{{{P_{i}}{h_{{x_i}}}{\left\|x_i\right\|}{^{ - \alpha_i 
}}}}{{\sum\limits_k {\sum\limits_{x \in {\Phi _k^{\textrm{BS}}}} 
{{P_k}{h_x}\left\|x\right\|{^{ - {\alpha _k}}}} }  + \sum\limits_k 
{\sum\limits_{y \in {\Phi _k^{\textrm{UE}}}} {{P_{\textrm{UE}}}} } {h_y}\left\|y\right\|{^{ - {\alpha 
_k}}} + \eta {P^\kappa } + {\sigma ^2}}}.\\ 
\label{eq7}
\end{align}
In this notation, the numerator represents 
the desired signal power received from a BS in the $i^{th}$ tier which serves 
the UE. Here, a tier defines the set of the BSs that have the same 
characteristics including average transmit powers, supported data rate, 
coverage areas, BSs density \cite{DhillonkHetNet}. 
The first and second term in the denominator represent the interference from other 
BSs and UEs in the network operating in the same RBs as the desired UE. Specifically, $\Phi_k^{\textrm{BS}}$ and $\Phi_k^{\textrm{UE}}$ indicate sets containing the 
positions of all \textit{interfering} BSs and UEs in the $k^{\textrm{th}}$ tier, and the summation is over all 
possible tiers. The third term is the RSI signal power as modeled 
in \eqref{eq}. Finally, $\sigma^{2}$ is the additive noise power. 

We recall that the number of bits delivered to a UE during an interval, $T_f$, in a given 
bandwidth, $BW$, if capacity achieving modulation and coding are used, can be 
represented as 

\begin{equation}
R=T_f\cdot BW\cdot\log_2\left(1+\textrm{SINR}\right)\textrm{.}
\label{eq8}
\end{equation}
Therefore, the effective capacity of the desired UE based on \eqref{eq5} can be 
expressed as 
\begin{align}
\textrm{EC}\left(\theta\right)&= 
-\frac{1}{\theta}\log\left(\mathbb{E}\left\{\textrm{exp}\left(-\theta\cdot T_f\cdot 
BW\cdot\log_2\left(1+\textrm{SINR}\right)\right)\right\}\right) \nonumber\\
\label{eq9}
& = 
-\frac{1}{\theta}\log\left(\mathbb{E}\left(\left(1+\textrm{SINR}\right)^{-\theta\cdot 
T_f\cdot BW\cdot\log_2e}\right)\right),
\end{align}
where the expectation is taken with respect to the SINR. 

In an HD scenario, a 1/2 scaling factor is needed and, also, the SINR would 
become
\begin{align} 
\textrm{SINR} = \frac{{{P_i}{h_{{x_i}}}{\left\|x_i\right\|}{^{ - \alpha_i 
			}}}}{{\sum\limits_k {\sum\limits_{x \in {\Phi _k^{\textrm{BS}}}} 
			{{P_k}{h_x}\left\|x\right\|{^{ - {\alpha _k}}}} } + {\sigma ^2}}}. 
\label{Dual-eq7}
\end{align}
\section{Theoretical analysis}
\label{sec:theoretical}
We aim at computing the QoS experienced by a generic UE, that can be placed 
anywhere in the coverage area of its own small cell with uniform distribution. 
To find the exact EC in a given topology, one needs to solve 
\eqref{eq9} either through extensive simulations or by mathematical analysis. 
It is worth mentioning that, if there are $M$ small cells within the macro 
cell, the associated integrals would be in a $2M+1$ dimensional parameter 
space, when a worst case scenario is assumed, i.e., any other small cell and 
the macro cell present one active UE operating in the same RB as our desired 
UE. On the other hand, if a simulation approach is pursued, the length of 
simulations 
in order to achieve a given confidence level will increase at least linearly 
with $M$.
This scaling may represent a prohibitive factor in finding the exact EC in a 
real scenario. 

\subsection{Approximating EC}
\label{subsection:approximating the EC}

Let us define a generic function $g$ of $s, I, a,$ and  $\beta$ as 
follows
\begin{equation}
g(s,I) =  \left(1+\frac{s}{I+a}\right)^{-\beta}.
\label{eq10}
\end{equation}
This function has the same structure of the expectation argument in 
\eqref{eq9}, where $s$ models the received signal power, $I$ represents the 
overall interference from other BSs and UEs in the network, $a$ represents the 
RSI and noise, and $\beta=\theta\cdot T_f\cdot BW\cdot\log_2 e$. Based on this 
definition, the EC can be rewritten as
\begin{equation}
EC(\theta)=-\frac{1}{\theta}\log \left( \E\limits_{s,I}  g \left( s,I  
\right)\right).
\label{eq11}
\end{equation}
\begin{Lemma}
\label{Lemma1}
For $0\le \beta \le 1$, $g$ is always a concave function of $I$. 
\end{Lemma}

\begin{proof}
By assuming $s$ as a constant, taking the second derivative of $g$ with respect 
to $I$ leads to 
\begin{align}
\frac{{{\partial ^2}g(s,I)}}{{{\partial}{I^2}}} = \frac{{\beta s}}{{{{\left( {I 
+ a} \right)}^4}}}{\left( {1 + \frac{s}{{I + a}}} \right)^{ - \left( {\beta  + 
2} \right)}}\left( { - 2\left( {I + a} \right) + \left( {\beta  - 1} \right)s} 
\right),
\label{eq12}
\end{align}
which is negative (i.e., $g$ is a concave function of $I$) for any value of 
$0\le \beta <1 + \frac{2}{{\textrm{SINR}}}$.
But since SINR is a random variable depending on the instantaneous 
values of signal power, $s$, and interference plus noise, $I+a$, we can only be 
sure that $g$ is always a concave function of $I$ for any $0\le \beta \le1$. 
\end{proof}
The concavity of $g$ helps to find a tight lower bound for the EC with greatly 
decreased complexity. To this end, by exploiting Jensen's inequality, we obtain
\begin{equation}
EC_{\textrm{LB}}\left( \theta  \right) =\frac{{ - 1}}{\theta }\log {\E_s}\left[ 
{{{\left( {1 + \frac{s}{{\bar{I}  + a}}} \right)}^{ - \beta }}} \right] 
\le EC\left( \theta  \right).
\label{eq13}
\end{equation} 

Here, $\bar{I}=\E_I(I)$ is the average interference experienced by the UE, and 
the remaining expectation only applies to the desired signal power. The 
advantage of this lower bound is its extremely reduced computational 
complexity. Indeed, calculating this lower bound only requires a 1--dimensional 
integral with respect to the desired signal power. 
Therefore, the proposed bound makes this calculation scalable with the size of 
the network, at the possible cost of losing some precision. 

It must be noted that the constraint on $\beta$, $0 \le \beta \le 1$, imposes a constraint on  
$\theta$, namely 
\begin{equation}
	0 \le \theta \le \frac{1}{T_f BW \log_2 e} \approx 10^{-2}.
	\label{constraintQoS}
\end{equation}
While this range generally includes all the meaningful values of the QoS 
exponent, in the simulations we will show that our method gives a good 
approximation for EC in an even wider range of $\theta$.

\begin{figure}[!t]
\vspace{-.3cm}
\centering
\begin{tikzpicture}
\draw[draw=black,thick,dotted] (-3,-.7) circle (.95);
\draw[draw=black,thick,dotted] (0,-.1) circle (1.3);
\draw[draw=black,fill=black] (-1.55,+1.15) circle (0.066);
\draw[draw=black,fill=black] (0,-.1) circle (0.06);
\draw[draw=black,fill=black] (-3,-.7) circle (0.06);
\draw[draw=black,thick,dashed] (-5.5,+1.5) arc (183:365:3.85cm);
\draw[draw=black,thick] (-2.81,-0.65) arc (10:130:0.2cm);
\draw[draw=black,thick] (0.25,-.04) arc (-10:95:0.2cm);
\draw[draw=black,thick] (0.01,.06) arc (90:170:0.2cm);
\draw[draw=black,thick] (-.8,-.12) arc (160:210:0.15cm);

\draw[latex-,draw=black,thick,dashed] (-5.45,+1.2) -- (-1.55,+1.15);	
\draw[-latex,draw=black,dashed] (-3,-.7) -- (-3,-1.65);	
\draw[-latex,draw=black,dashed] (0,-.1) -- (0.25,-1.355);
\draw[line width=0.7pt, draw=black] (-3,-.7) -- (0.5,.02);
\draw[-latex, draw=black] (-3,-.7) -- (-3.4,-.2);
\draw[-latex, draw=black] (0,-.1) -- (0.1,0.65);
\draw[draw=black] (0.1,0.65) -- (-3.4,-.2);
\draw[draw=black] (-3.4,-.2) -- (0,-.1);

\node[label=below:\rotatebox{0}{$R_1$}] at (-3.15,-.8) {};
\node[label=below:\rotatebox{0}{$R_2$}] at (+0.3,-.405) {};
\node[label=below:\rotatebox{0}{$r_1$}] at (-3.35,-.25) {};
\node[label=below:\rotatebox{0}{$\theta_1$}] at (-2.8,0) {}; 
\node[label=below:\rotatebox{0}{$r_2$}] at (-.11,0.7) {};
\node[label=below:\rotatebox{0}{$\theta_2$}] at (0.41,0.6) {};  
\node[label=below:\rotatebox{0}{$d$}] at (-1.75,-0.23) {};
\node[label=below:\rotatebox{0}{$c$}] at (-1.75,.28) {};
\node[label=below:\rotatebox{0}{$x$}] at (-1.75,0.65) {};
\node[label=below:\rotatebox{0}{$\gamma$}] at (-.25,0.35) {};
\node[label=below:\rotatebox{0}{$\psi$}] at (-1.1,0.2) {};

\node at (-1.55,+1.35) {Macro BSs};
\end{tikzpicture}
%\vspace{-.25cm}
\caption{Structure of interference in the system.}
\label{fig2}
\end{figure}
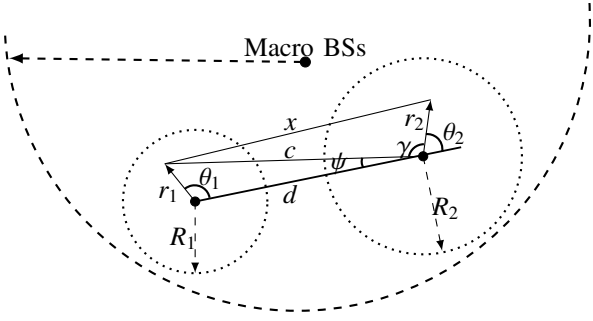

\subsection{Average Interference on a UE}
\label{subsec:average_interference}

In order to efficiently compute the lower bound 
$EC_{\textrm{LB}}\left(\theta\right)$, one has to calculate analytically the 
average interference, $\bar {I} $. 
We propose here a mathematical analysis that could serve as a building 
block for this goal. 

Recalling the expression of the interference from 
Section~\ref{subsec:interference_modelling}, the average interference from an interferer located at distance $x$ from the considered UE can be found as
\begin{equation}
\E[Ph\Vert x\Vert^{-\alpha}] \mathop  = \limits^{\rm{(a)}}  P\E (h) \E[\Vert 
x\Vert^{-\alpha}] \mathop  = \limits^{\rm{(b)}} P \E[\Vert x\Vert^{-\alpha}],
\label{eq14}
\end{equation}
where (a) follows from the fact that the channel coefficient and distance 
between the interferer and the desired UE are independent random variables and 
(b) holds because the random variable $h$, which accounts for channel fading, 
has unit mean. Consequently, all our efforts will be focused on finding the 
average path loss from the desired UE to the interferers. 

Fig. \ref{fig2} depicts a deployment of two small cells and their 
corresponding coverage areas within a macro cell. This is an example of a 
Mat\'ern HCPP with two cluster heads (BSs) and a hard core distance $ r_h \ge 
{R_1+R_2} $. This assumption for the hard core distance makes the two small 
cells non--overlapping. The probability density functions (PDFs) of the 
interferer and desired UE locations, expressed in polar coordinates, are, 
respectively, $f_i(r_i,\theta_i)=\frac{r_i}{\pi R_i^2}$, $i={1,2}$ within their 
coverage area and zero outside. 

We can find the squared distance between the desired UE and the interferer as
\begin{equation}
{x^2} = {c^2} + r_2^2 - 2c{r_2}\cos \left( \gamma  \right),
\label{eq15}
\end{equation}
where ${c^2} = r_1^2 + {d^2} - 2{r_1}d\cos \left( {{\theta _1}} \right)$, 
$\gamma \,\, = \pi - \theta _2 - \psi$ and $\psi$ is also a random variable 
depending on the interferer's position.

Our goal is to compute the average path loss between the considered UE and the interferers
\begin{equation} 
\E\left[ {{{\left\| x \right\|}^{ - \alpha }}} \right] = \E \left[ {{{\left( 
{{{ {{c^2} + r_2^2 + 2c{r_2}\cos \left( {{\theta _2} + \psi } \right)} }}} 
\right)}^{ - \frac{\alpha }{2}}}} \right]\ ,
\label{eq16}
\end{equation}
which is challenging to compute in general. 
To this end, we first compute the expectation in \eqref{eq16} by assuming a 
fixed position for the interferer, i.e., fixed $(r_1,\theta_1)$. Subsequently, 
we compute the expectation of the resulting quantity with respect to all possible 
values of $(r_1,\theta_1)$. 

Regarding the first step, since we assumed $(r_1,\theta_1)$ is fixed, 
$c$ and $\psi$ become constants, thus facilitating the analysis
\begin{eqnarray}
\label{fixedr1theta1}
&&\E\limits_{\left( {{r_2},{\theta _2}} \right)} \left[ {\left. {{{\left( 
{{c^2} + r_2^2 + 2c{r_2}\cos \left( {{\theta _2} + \psi } \right)} \right)}^{ - 
\frac{\alpha }{2}}}} \right|\left( {{r_1},{\theta _1}} \right)} \right]\\ \nonumber
&&= \int\limits_0^{2\pi } {\int\limits_0^{R_2} {{c^{ - \alpha }}{{\left( {1 + 
{{\left( {\frac{{{r_2}}}{c}} \right)}^2} + 2\left( {\frac{{{r_2}}}{c}} 
\right)\cos \left( {{\theta _2} + \psi } \right)} \right)}^{ - \frac{\alpha 
}{2}}}} } .\frac{1}{\pi }\frac{{{r_2}}}{{{R_2^2}}}{\mkern 1mu} {\mkern 1mu} 
d{r_2}d{\theta _2}\\ \vspace{-1cm}
\\ \label{Closed_Form}
&&\mathop  \simeq \limits^{(a)} {c^{ - \alpha }}\left[ {1 + \frac{{{{\alpha ^2}}}}{8}\left( 
{\frac{{{R_2^4}}}{{3{c^4}}} + \frac{{{R_2^2}}}{{{c^2}}}} \right)+\frac{{ \alpha 
}}{4} {\frac{{{R_2^4}}}{{3{c^4}}}}} \right] \\ \nonumber
\end{eqnarray}
where in (a) we used the first three terms of the Taylor series expansion of 
${\left( {1 + x} \right)^{ - \omega}} = 1 - \omega x + \frac{{\omega\left( 
{\omega + 1}\right)}}{2!}{x^2} +\dots $. 
The Taylor approximation is legitimate if $x<1$ $(c>R_2)$ which is already 
satisfied considering the repulsive point process we have assumed for the small 
cells, characterized by the hard core distance $r_h\ge{R_1+R_2}$.

We recall that this result holds for any fixed values of 
$\left(r_1,\theta_1\right)$. In particular, by setting $r_1 \rightarrow 0$ 
(i.e., $c \rightarrow d$ in \eqref{Closed_Form}), we obtain the average path 
loss 
component between a randomly deployed UE and the BS. Therefore, the average 
interference that an external BS causes to the considered UE uniformly placed 
in any point within the coverage area of its small cell is
\begin{equation}
{I_{{\rm{BS - UE}}}}=P_{BS}{d^{ - \alpha }}\left[ {1 + \frac{{{{\alpha 
^2}}}}{8}\left( {\frac{{{R_2^4}}}{{3{d^4}}} + \frac{{{R_2^2}}}{{{d^2}}}} 
\right)+\frac{{ \alpha }}{4} {\frac{{{R_2^4}}}{{3{d^4}}}}} \right].
\label{eq18}
\end{equation}
To find the average interference generated by another UE we have to 
take the expectation of \eqref{fixedr1theta1} with respect to $(r_1,\theta_1)$ 
\begin{align}
\label{eq19.1}
&\,\,{I_{{\textrm{UE - UE}}}} = {P_{UE}} \cdot \\ \nonumber
&\left( {\mathop {\E} \limits_{\left( {{r_1},{\theta _1}} \right)} \left[ {{c^{ 
- \alpha }}} \right] + \frac{{{\alpha ^2}{R_2^2}}}{8}\mathop {\E} 
\limits_{\left( {{r_1},{\theta _1}} \right)} \left[ {{c^{ - (\alpha  + 2)}}} 
\right] + \frac{{\alpha \left( {\alpha  + 2} \right){R_2^4}}}{{24}}\mathop {\E} 
\limits_{\left( {{r_1},{\theta _1}} \right)} \left[ {{c^{ - (\alpha  + 4)}}} 
\right]} \right).
\label{eq19.2}
\end{align}
To compute the quantity in \eqref{eq19.1}, one needs to calculate only 
$\E\limits_{\left({{r_1},{\theta _1}} \right)}\left[c^{-\alpha}\right]$, since 
the other two expectations can be immediately obtained by replacing $\alpha$ 
with $\alpha+2$ and $\alpha+4$.
We further observe that this expectation corresponds to the average path loss 
component between the interferer and the desired UE's BS. This quantity can be 
derived from \eqref{eq18} by setting $P_{BS}=1$ and changing $R_2$ to 
$R_1$.

The proposed relations in \eqref{eq18} and \eqref{eq19.1} can hence be used as 
a basic mathematical tool to investigate the system performance. 

\section{Simulations and Results}
\label{sec:simulations}
In the simulations we considered a single cell scenario where the macro BS 
is located at the center, overlaid with randomly placed small cells. The 
simulation parameters are reported in Tab.~\ref{tab:parameters}. The
system for the HD scenario is assumed to be frequency--division duplexing 
(FDD). In addition, we should note that the only source of interference that 
does not follow 
the structure provided in Fig. \ref{fig2} is the UE connected to the macro BS. 
For this specific UE we assume that the network operator grants different RBs 
in UL transmission compared to our desired UE (in other words, we assume this 
UE is in HD mode of operation).

The results for the network realizations reported in Figs. 3.a and 4.a are 
shown in Figs.~\ref{fig.4}.b and ~\ref{fig.6}.b, respectively. The dashed small 
cell is the one under investigation and UEs are uniformly deployed around their 
corresponding BSs'. The curves of the EC perceived by a typical UE in the 
dashed small cell is computed according to different methods, for both cases of 
an HD and an entirely FD system.
Specifically, the exact EC for HD and FD (red curves) is obtained by simulating 
\eqref{eq11} in the given HCN realization for randomly placed UEs in the network while the analytical--simulation 
results (green curves) are based on the lower bound provided in \eqref{eq13} 
where the average interference is calculated by using the relations given in 
\eqref{eq18} and \eqref{eq19.1} and the remaining expectation with respect to 
the signal power is obtained through simulation. Finally, to validate our 
analytical calculation of the average 
interference on the desired UE, we plotted the lower bound in 
\eqref{eq13} obtained by computing the average interference through 
simulation rather than using our theoretical analysis (LB--Simulation black 
curves).

\begin{table}[t]
\center
\caption{System and Simulation Parameters}
\begin{tabular}{lll}
	\hline
	\cmidrule(r){1-3}
	Description    & Parameter & Value \\
	\midrule
    Macro BS TX Power       & $P_{\text {M-BS}}$         &  46 dBm      \\
    Pico BS TX Power        & $P_{\text {P-BS}}$         &  35 dBm      \\
    User TX Power            & $P_{\text {UE}}$             &  23 dBm      \\
    Path loss exponent       & $\alpha$                     &  3           \\
    Noise Power              & $ \sigma^2 $                 &  -120 dBm    \\
    Pico--Pico BSs Minimum Distance  &-                     &  180 meters    \\                
    Coverage Radii of Pico cells               &- 
    & 90 meters \\
	\bottomrule
\end{tabular}
\label{tab:parameters}
\end{table}

\begin{figure*}[!t]
	\centering
	\subfigure[]{\includegraphics[width=0.4 \textwidth]{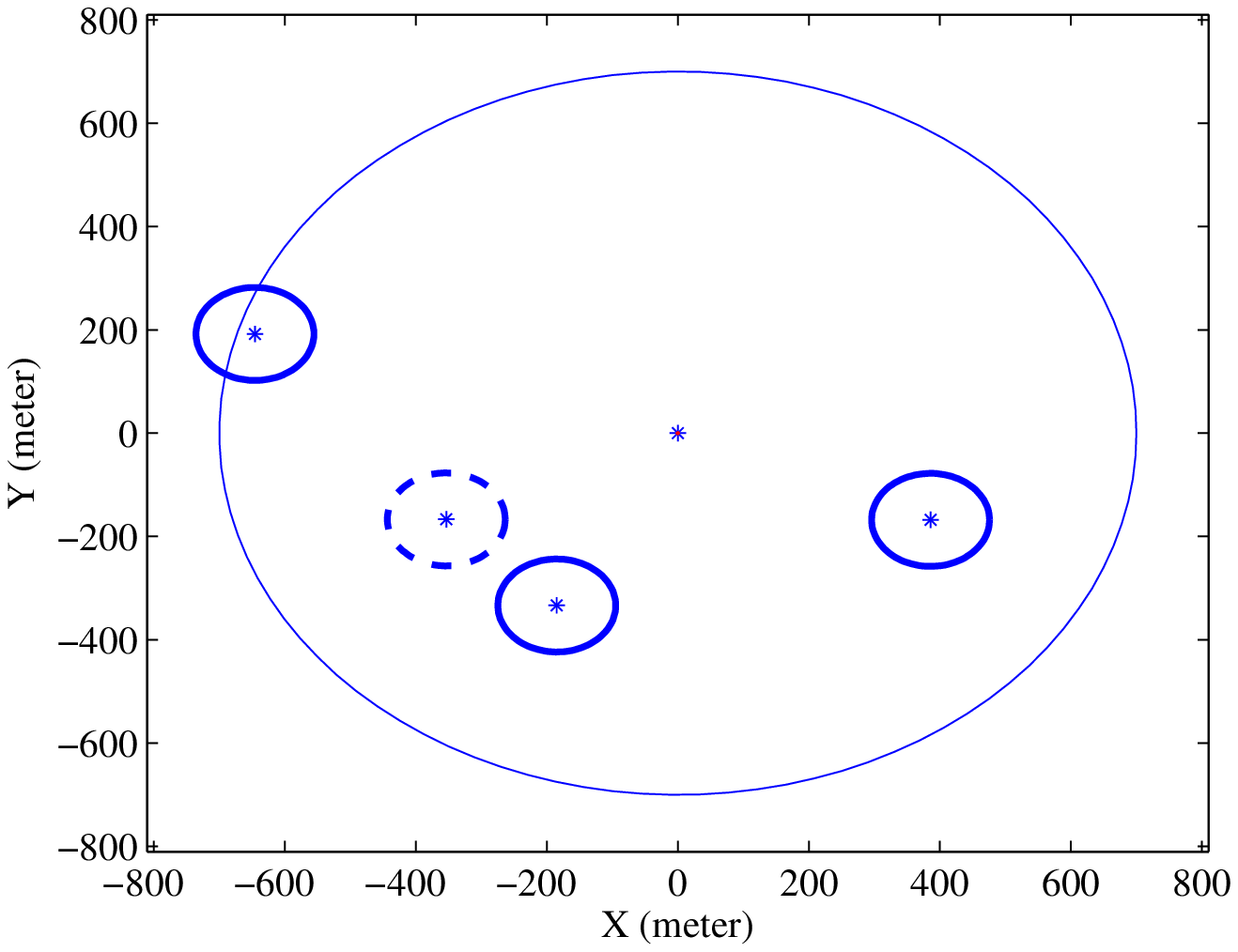}}
	\hfil
	\subfigure[]{\includegraphics[width=0.4 \textwidth]{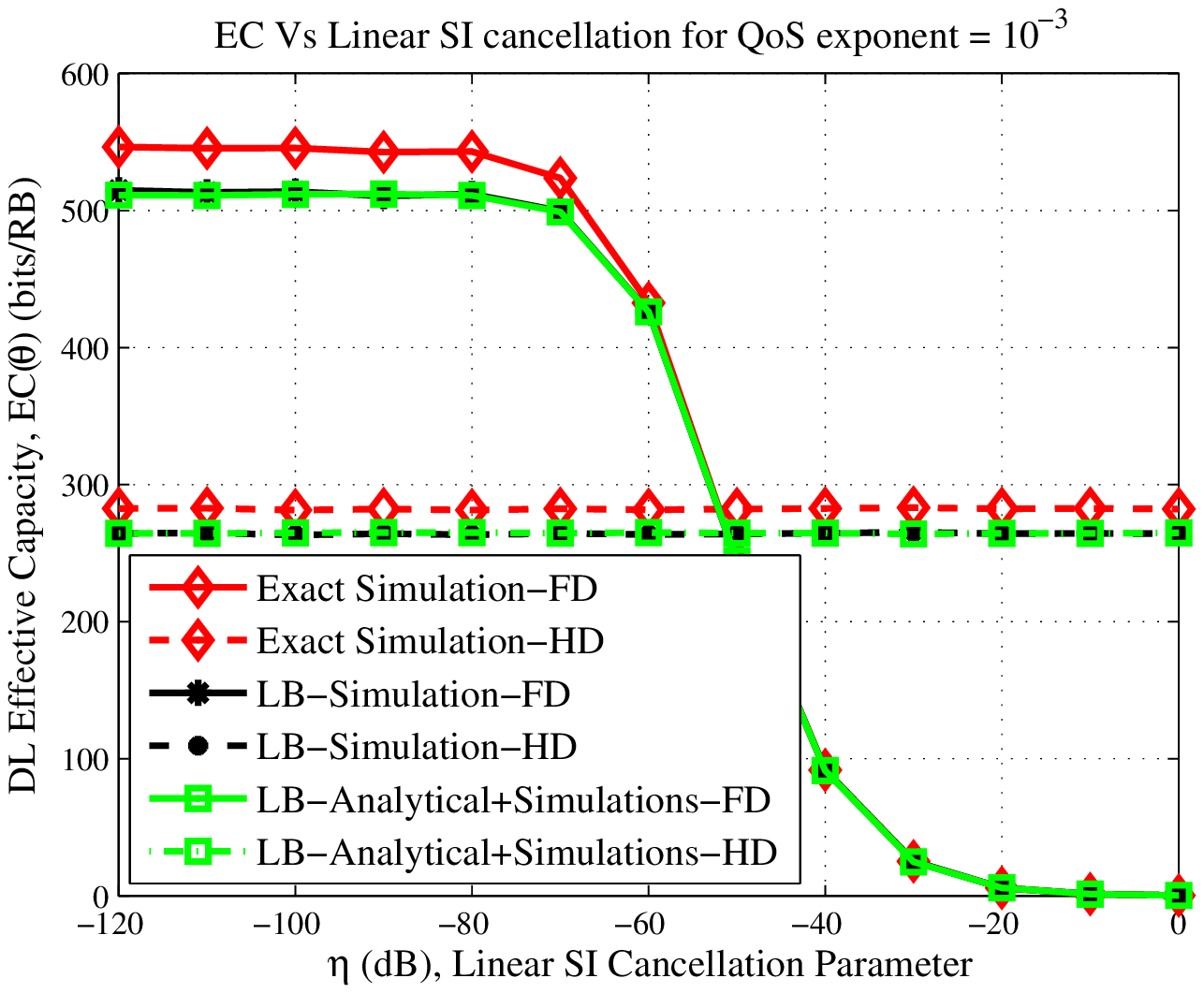}}
	\vspace{-.35cm}
	\caption{A specific instance of sparse small cell deployment as obtained 
	using a Hard Core Mat\'ern PP with density $\lambda=5 \, 
		\textrm{small cell/km}$${^2}$ (a), and corresponding DL effective 
		capacity experienced by a typical UE (uniformly distributed in the 
		dashed small--cell) vs. linear self--interference 
		suppression ratio, for HD and FD (exact and lower bounds) (b). QoS 
		exponent $\theta=10^{-3}$ (1/bit), non linear self--interference 
		cancellation parameter $\kappa=1$.}
	\label{fig.4}
\end{figure*}
\begin{figure*}[!t]
	\vspace{-.55cm}
	\centering
	\subfigure[]{\includegraphics[width=0.4 \textwidth]{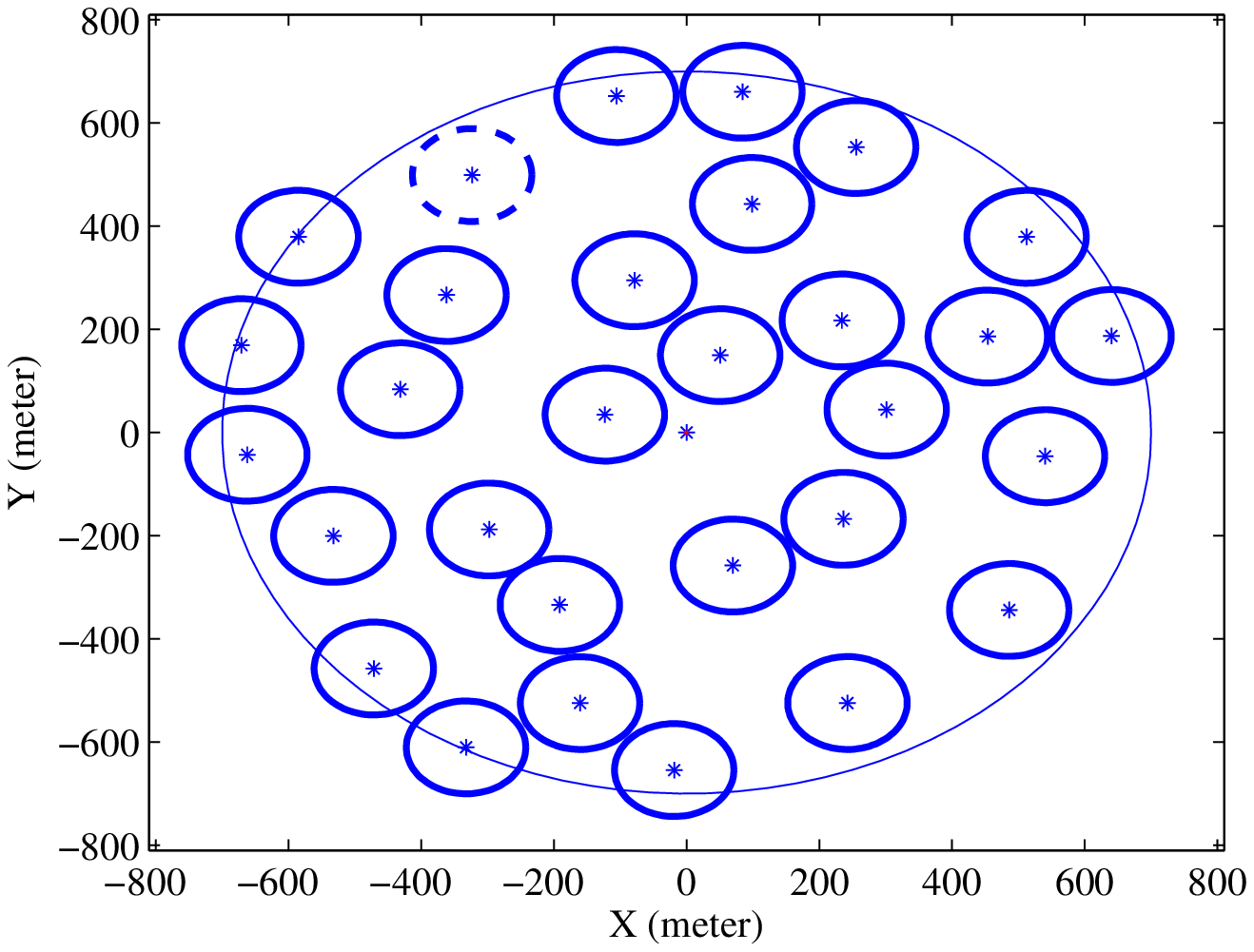}}
	\label{fig.8}
	\hfil
	\subfigure[]{\includegraphics[width=0.4
		\textwidth]{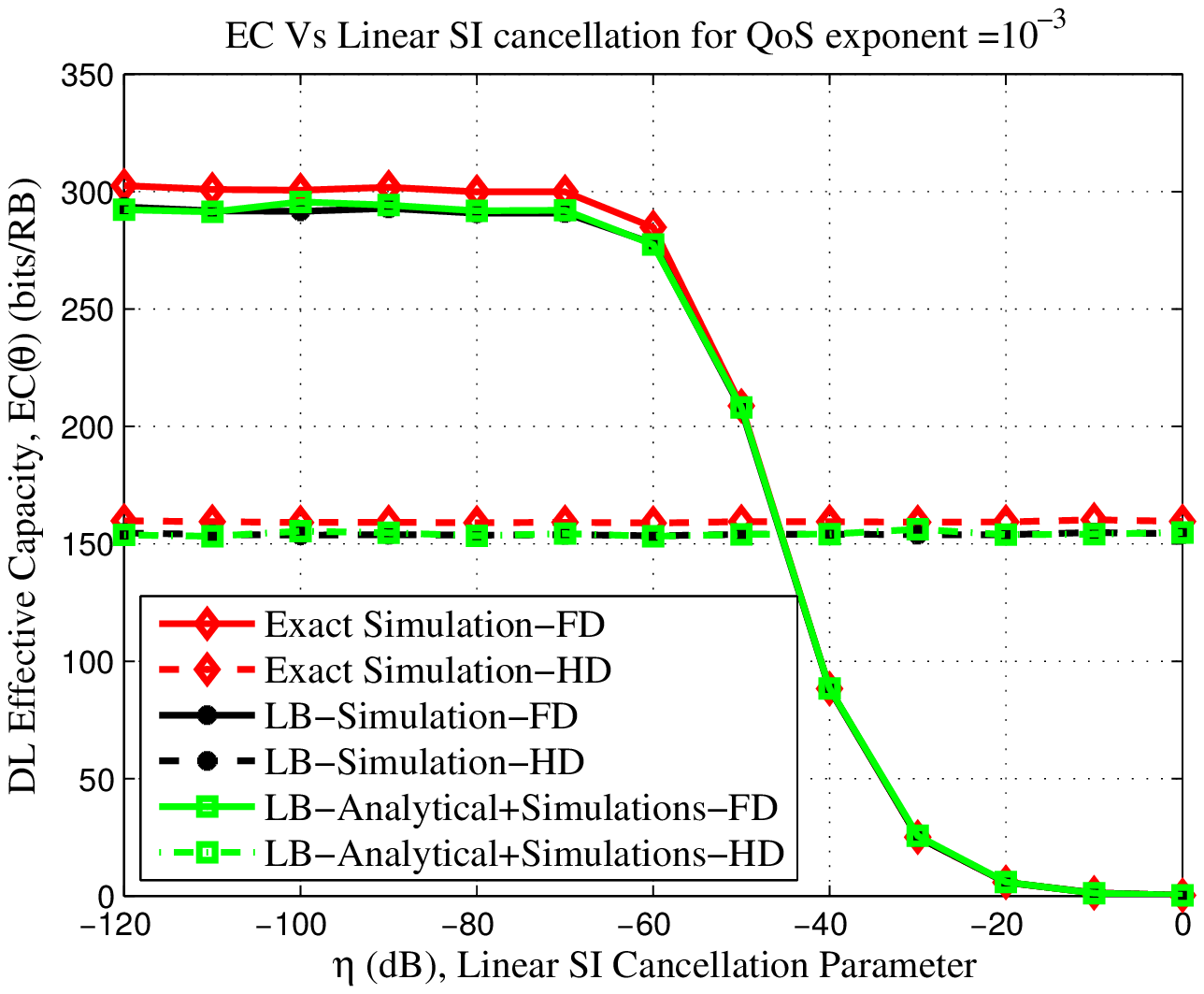}}
	\vspace{-.35cm} 
	\caption{A specific instance of dense small cell deployment as obtained 
	using a Hard Core Mat\'ern PP with density $\lambda=50 \, 
		\textrm{small cell/km}$$^2$ (a), and corresponding DL effective 
		capacity experienced by a typical UE (uniformly distributed in the 
		dashed small--cell) vs. linear self--interference 
		suppression ratio, for HD and FD (exact and lower bounds) (b). QoS 
		exponent $\theta=10^{-3} (1/bit)$, non linear self--interference 
		cancellation parameter $\kappa=1$.}
	\label{fig.6}
\end{figure*}

Fig.~\ref{fig.4} refers to a sparse system, with $\lambda=5 $ {small 
cells per km$^2$}, and reports the downlink (DL) effective capacity versus the
linear SI cancellation parameter, for QoS exponent $\theta=10^{-3}$. 
For the given QoS exponent, it can be inferred that a maximum gain of 1.93X can 
be achieved with the help of a perfect FD system. A similar maximum gain of 
1.89X is reported in Fig.~\ref{fig.6}, this time for the realization of a denser
Mat\'ern HCPP with $\lambda=50$ small cells per km$^2$ and a QoS exponent 
$\theta=10^{-3}$. 
In both cases, a trade--off value for the linear SI cancellation 
parameter at which the FD operation mode outperforms HD in terms of downlink EC 
can be found, namely $-50$~dB for the former scenario and $-45$~dB for the 
latter. 
Moreover, most of the maximum FD gain obtainable can already be achieved for 
$\eta=-80$~dB in the first scenario and $\eta=-70$~dB for the second one, which 
are readily provided by current technology.
In both scenarios, the non linear SI cancellation parameter was set to 
$\kappa=1$.

The second important result that can be observed from these figures is the fact 
that the lower bound proposed in \eqref{eq13} is tight. Specifically, the black 
curves, for the lower bound computed through simulations, and the green ones, 
for the lower bound computed through analysis and simulations, are practically 
overlapped and very close to the red curves that report the exact value of EC. 

It is worth observing that the lower bound is closer to the exact values if 
the system becomes more crowded, i.e., for a higher density of BSs. 
In addition, as tabulated in Table \ref{tabel2} and discussed in Section 
\ref{sec:theoretical}, the analytical approach has a complexity almost 
independent of the network size and significantly lower with respect to the exact computation of EC. Thus, our method to analyze statistical QoS 
performance of HCN is scalable with the network size.
\begin{table}[t]
	\vspace{-.25cm}
	\center
	\caption[Time Elapsed in Analyzing the System Performance]{Time Elapsed in 
	Analyzing the System Performance\footnotemark{}}
	\begin{tabular}{ccc}
		\hline
		\cmidrule(r){1-3}
		Scenario    & Exact Analysis  & Proposed Lower Bound \\
		\midrule
		Fig. \ref{fig.4}        & 370 s            &  17 s     \\
		Fig. \ref{fig.6}       & 2220 s           &   21 s      \\
		\bottomrule
	\end{tabular}
	\label{tabel2}
\end{table} 
\footnotetext{All the simulations are carried out with an Intel Core i5-2.53GHz processor and 4G RAM on a Dell Inspiron 5010.}
%\vspace{-.25cm}
\section{Conclusions}
\label{sec:conclusions}
In this paper we introduced a lower bound for the evaluation of the effective 
capacity in a generic wireless scenario. Based on the proposed lower bound we 
built a scalable mathematical framework to analyze the statistical QoS 
performance of dense next generation HCNs. Our proposed scheme helped us 
analyze HD and imperfect FD HCNs from an EC perspective with very good accuracy 
at only a fraction of the complexity needed for an exact analysis.
%As a possible extension of this work, the EC approximation can be improved by 
%considering also the case of FD UEs directly associated to the macro cell. 
%Moreover, the precision of the approximation in the case of partially 
%overlapping small cells can be investigated.

%\vspace{-.25cm}
\bibliographystyle{IEEEtran}
\bibliography{References}

\end{document}